\documentclass[twocolumn,showkeys]{article}
\usepackage[pdftex]{graphicx,xcolor}
\usepackage{graphicx, color}
\usepackage{subcaption}
\usepackage{amsmath}
\usepackage{amssymb,amsfonts,mathtools,amsthm}
\usepackage{authblk}
\usepackage{comment}
\makeatletter
\@addtoreset{equation}{section}

\makeatother
\newtheorem{prop}{Proposition}
\usepackage[top=20truemm,bottom=20truemm,left=20truemm,right=20truemm]{geometry}
\title{Stability analysis and control of decision-making of miners in blockchain}
\author[1]{Kosuke Toda}
\author[1]{Naomi Kuze}
\author[1]{Toshimitsu Ushio}
\affil[1]{Graduate School of Engineering Science, Osaka University, Machikaneyama 1-3, Toyonaka-shi, Osaka, 560--8531, Japan.}
\date{}

\begin{document}
\maketitle
\begin{abstract}
  To maintain blockchain-based services with ensuring its security, it is an important issue how to decide a mining reward so that the number of miners participating in the mining increases.
  We propose a dynamical model of decision-making for miners using an evolutionary game approach and analyze the stability of equilibrium points of the proposed model.
  The proposed model is described by the 1st-order differential equation.
  So, it is simple but its theoretical analysis gives an insight into the characteristics of the decision-making.
  Through the analysis of the equilibrium points, we show the transcritical bifurcations and hysteresis phenomena of the equilibrium points.
  We also design a controller that determines the mining reward based on the number of participating miners to stabilize the state where all miners participate in the mining.
  Numerical simulation shows that there is a trade-off in the choice of the design parameters.
\end{abstract} \hspace{10pt}


\section{Introduction}
Blockchain is a distributed ledger technology for recording transactions that underlies various fields such as digital currency like Bitcoin~\cite{nakamoto2008bitcoin}, data sharing~\cite{Xia2017}, and computer security~\cite{Ouaddah2016}.
Blockchain-based services use cryptography to record transactions as a chain of blocks.
A block consists of a block header and transaction data.
The block header contains a cryptographic hash of its previous block, which makes blockchain-based services resistant to tampering.
In these services, participants called \textit{miners} create blocks in a distributed manner, and the longest chain of blocks is considered to be legitimate.
When a miner succeeds in creating a block, he/she gets a reward called a \textit{mining reward}.

Blockchain-based services approve transactions through a consensus algorithm.
As a consensus algorithm, proof-of-work (PoW) is typically used.
In this algorithm, the mining difficulty is set using a scalar value called a \textit{nonce} in the block header.
To create a block, miners must find a nonce such that the cryptographic hash value for the previous block satisfies specific conditions.
The process of creating blocks is called a \textit{mining}.
In general, a cryptographic hash value for a block is unique according to the nonce contained in the block.
Moreover, a nonce that satisfies the specific conditions cannot be calculated directly.
As a result, an exhaustive search imposes a large computational cost on miners, which contributes to the resistance to tampering.
Because transaction approvals depend on miner calculations (such calculations are very costly and require a lot of energy~\cite{TRUBY2018399,Cambridge}),
the participation of many miners is needed to maintain blockchain-based services and ensure blockchain system security~\cite{Cai2018,liu2021social}.
Therefore, it is important to analyze the decision-making problem of whether miners participate in the mining according to the energy consumption and mining rewards.

Game theory is used to analyze interactions among rational decision-makers.
Many studies have adopted game theory to analyze blockchain-related issues with PoW~\cite{Liu2019}, such as decision-making problems in the mining considering the energy consumption~\cite{Dimitri2017,Fiat2019}.
Evolutionary game theory has been used as a powerful mathematical tool for analyzing dynamical models of evolutionary selection~\cite{weibull1997evolutionary}.
Dynamical characteristics of the selection process are modeled by \textit{replicator dynamics}.
Control methods for the replicator dynamics have been studied in~\cite{Kanazawa2008,Kanazawa2009,Morimoto2016}.
Evolutionary game models and replicator dynamics are also used in analyzing blockchain-related issues such as mining pool selection problems~\cite{Liu2017,Fujita2020}, and attack scenarios~\cite{Kim2019}.

We previously focused on a decision-making problem of whether miners participate in the mining according to the energy consumption and the mining rewards, and modeled it as a non-cooperative game.
Through theoretical and numerical analysis, we showed the property of Nash equilibria~\cite{Toda2021}.
However, in this study, we assumed that, once miners choose a strategy~(i.e., participation in the mining or not), they do not change their strategies.
Practically, the miners may decide to participate in the mining dynamically based on their current earned mining rewards.
It is important to analyze such a dynamical decision-making process.

In this paper, we propose a dynamical model of the decision-making problem for miners, by applying an evolutionary game approach.
We analyze the stability of its equilibrium points and show the existence of transcritical bifurcations and hysteresis phenomena with the coexistence of two asymptotically stable equilibrium points: one corresponds to the state where all miners participate in the mining and the other to the state where the number of participating miners is minimum.
The former equilibrium point is preferable to maintain blockchain-based services.
We propose a controller that determines the mining reward based on the number of current participating miners so as to stabilize the equilibrium point if at least one miner participates in the initial time.

The remainder of this paper is organized as follows.
In Section~\ref{sec:formulation}, we propose an evolutionary game-based dynamical model of the decision-making process.
In Section~\ref{sec:analysis}, we analyze the stability of its equilibrium point.
In Section~\ref{sec:stabilize}, we design a state feedback controller to let all miners participate in the mining.

\section{Dynamical model of decision-making}\label{sec:formulation}
We assume that miners in a blockchain network are partitioned into two sets $\mathcal{M}$ and $\mathcal{N}$, where miners in $\mathcal{M}$ always participate in the mining
and those in $\mathcal{N}$ have two strategies, participating in the mining (strategy $s_k = 1$) and not participating in the mining (strategy $s_k = 0$), where $k \in \mathcal{N}$.
Note that $\mathcal{M} \cap \mathcal{N} = \emptyset$. Denoted by $m$ and $n$ are the cardinalities of $\mathcal{M}$ and $\mathcal{N}$, respectively (we assume $m \geq 1$ and $n \geq 1$).
We define $x_0$ and $x_1$ as the ratios of miners in $\mathcal{N}$ that choose strategies $0$ and $1$, respectively.
Note that
\begin{align}
  x_0 + x_1 = 1. \label{eq:relation_x0x1}
\end{align}

Miners need to find a nonce such that the first $h$ bits of the hash of the block are all $0$.
Then, $D = 2^h$ is the difficulty parameter and $1/D$ is the probability that a miner creates a block with one hash calculation~\cite{Debus2017}.
When miner $k \in \mathcal{M} \cup \mathcal{N}$ participates in the mining, he/she needs a cost $c$ per unit operating time.
The average number $w_k = f_k(c)$ of hash queries calculated per unit operating time by miner $k$ depends on the cost $c$, and we assume that $f_k(c)$ is the same for all miners.
In this paper, for simplicity, we assume $f_k(c) = vc \, (v > 0)$ for any $k \in \mathcal{M} \cup \mathcal{N}$.

The mining of blocks can be described as a Poisson process~\cite{Houy2016,Kraft2016}.
That is, the block creation time is exponentially distributed~\cite{1547-5816_2019_1_365}.
The rate $\lambda_k$ of the Poisson process of miner $k \in \mathcal{M} \cup \mathcal{N}$ is given by $\lambda_k = w_k/D$~\cite{Kraft2016}\footnote{
Note that a combination of independent Poisson processes is still a Poisson process. Thus, the rate of the Poisson process of all miners is written as $\sum_{i \in \mathcal{M} \cup \mathcal{N}} \lambda_i$.
}.
If miner $k$ chooses $s_k= 1$, then the rate of the Poisson process is
$\lambda_k = s_k f_k(c)/D = s_kc/d$ (we define $d = D/v$, in this paper).
Let $R$ be the mining reward.
Based on the previous work~\cite{Toda2021}, the expected reward $R_k$ and the expected cost $CS_k$ for the mining of miner $k$ are calculated as follows.
\begin{align}
  &R_k = \frac{\lambda_k}{\sum_{i \in \mathcal{M} \cup \mathcal{N}} \lambda_i} R = \frac{Rs_k}{m + nx_1}, \label{eq:expected_reward} \\
  &CS_k = \frac{c \lambda_k}{(\sum_{i \in \mathcal{M} \cup \mathcal{N}} \lambda_i)^2} = \frac{ds_k}{(m + nx_1)^2}. \label{eq:expected_cost}
\end{align}

We define the utility function $u_i(x_0, x_1)$ of miners that choose the strategy $i \in \{0, \ 1\}$ as
\begin{align}
  u_i(x_0, x_1) = \begin{cases}
    0 & \mbox{if} \; \; i = 0, \\
    \frac{1}{m + nx_1} \left(R - \frac{d}{m + nx_1}\right) & \mbox{if} \; \; i = 1,
\end{cases} \label{eq:utility}
\end{align}
which means that the utility of a miner who participates in the mining is the difference between the expected reward $R_k$ and the expected cost $CS_k$.
Based on the principle of the evolutionary game~\cite{weibull1997evolutionary}, the dynamics of the ratio of miners that choose the strategy $i$ is given by
\begin{align}
  \frac{\dot{x}_i}{x_i} = u_i(x_0, x_1) - \bar{u}(x_0, x_1) \; (i = 0, 1), \label{eq:replicator}
\end{align}
where $\bar{u}(x_0, x_1) = \sum_{i = 0}^1 x_i u_i(x_0, x_1)$ is the average utility of all miners. According to \eqref{eq:utility}, \eqref{eq:replicator} is rewritten as
\begin{align}
  \dot{x}_1 = -\dot{x}_0 = \frac{x_1(1 - x_1)}{m + nx_1} \left(R - \frac{d}{m + nx_1}\right) \eqqcolon \varphi_R(x_1). \label{eq:dynamics_x1}
\end{align}

Thus, the dynamics of the decision-making of miners is described by the above 1st-order differential equation and the reward $R$ plays an important role in the decision-making of the miners for the participation in the mining.
In the following, we investigate stability and stabilization of equilibrium points of~\eqref{eq:dynamics_x1}.
For that purpose, the concept of a \textit{basin of attraction}~\cite{khalil2002nonlinear} is important.
Let $\xi(t; x_1^{\mathrm{init}})$ be the solution of \eqref{eq:dynamics_x1} that starts from an initial state $x_1^{\mathrm{init}}$ at time $t = 0$.
For a given asymptotically stable equilibrium point $x_1^{\prime}$ of \eqref{eq:dynamics_x1}, the basin of attraction is defined as the set of all initial states $x_1^{\mathrm{init}}$
such that $\xi(t; x_1^{\mathrm{init}})$ is defined for all $t \geq 0$ and $\lim_{t \to \infty} \xi(t; x_1^{\mathrm{init}}) = x_1^{\prime}$.

\section{Stability analysis}\label{sec:analysis}
In this section, we investigate the stability of the equilibrium point $x_1 = 0, \ 1, \ x_1^{*}$ of \eqref{eq:dynamics_x1}, where
\begin{align}
  x_1^{*} = \frac{1}{n} \left(\frac{d}{R} - m\right). \label{eq:xstr}
\end{align}
When $1/(m + n) < R/d < 1/m$, the equilibrium point $x_1^{*}$ satisfies $0 < x_1^{*} < 1$. This equilibrium point is the state where the utility $u_1(1 - x_1^{*}, x_1^{*})$ is equal to $0$, that is, the utility for the strategy $0$ is equal to that for the strategy $1$.

We investigate the local stability of the three equilibrium points.
The derivative of $\varphi_R(x_1)$ with respect to $x_1$ is
\begin{align}
  &\frac{\partial \varphi_R(x_1)}{\partial x_1} = \left(-\frac{x_1}{m + nx_1} + \frac{1 - x_1}{m + nx_1} - \frac{nx_1(1 - x_1)}{(m + nx_1)^2}\right) \nonumber \\
  &\hspace{20mm}\times \left(R - \frac{d}{m + nx_1}\right) + \frac{dnx_1(1 - x_1)}{(m + nx_1)^3}. \label{eq:partial_phi}
\end{align}
Thus, we obtain
\begin{align}
  &\left.\frac{\partial \varphi_R(x_1)}{\partial x_1}\right|_{x_1 = 0} = \frac{1}{m} \left(R - \frac{d}{m}\right), \label{eq:dphi_0} \\
  &\left.\frac{\partial \varphi_R(x_1)}{\partial x_1}\right|_{x_1 = 1} = -\frac{1}{m + n} \left(R - \frac{d}{m + n}\right), \label{eq:dphi_1} \\
  &\left.\frac{\partial \varphi_R(x_1)}{\partial x_1}\right|_{x_1 = x_1^{*}} = \frac{R^3}{nd^2} \left(\frac{d}{R} - m\right) \left((m + n) - \frac{d}{R}\right). \label{eq:dphi_xstr}
\end{align}
\begin{table}[t]
  \centering
  \caption{The relation between $R/d$ and the stability of equilibrium points.}
  \begin{tabular}{|c|c|c|c|}\hline
    Condition for $R/d$ & $x_1 = 0$ & $x_1 = 1$ & $x_1 = x_1^{*}$ \\ \hline
    $R/d < 1 / (m + n)$ & S & U & S \\
    $1/(m + n) < R/d < 1 / m$ & S & S & U \\
    $R/d > 1 / m$ & U & S & S \\\hline
  \end{tabular}
  \label{tab:stability}
\end{table}
Thus, we have their stability conditions as shown in Table~\ref{tab:stability},
where S (\textit{resp.} U) represents an asymptotically stable (\textit{resp.} unstable) point.

\begin{figure}[t]
  \centering
  \includegraphics[clip, width = 7.5cm]{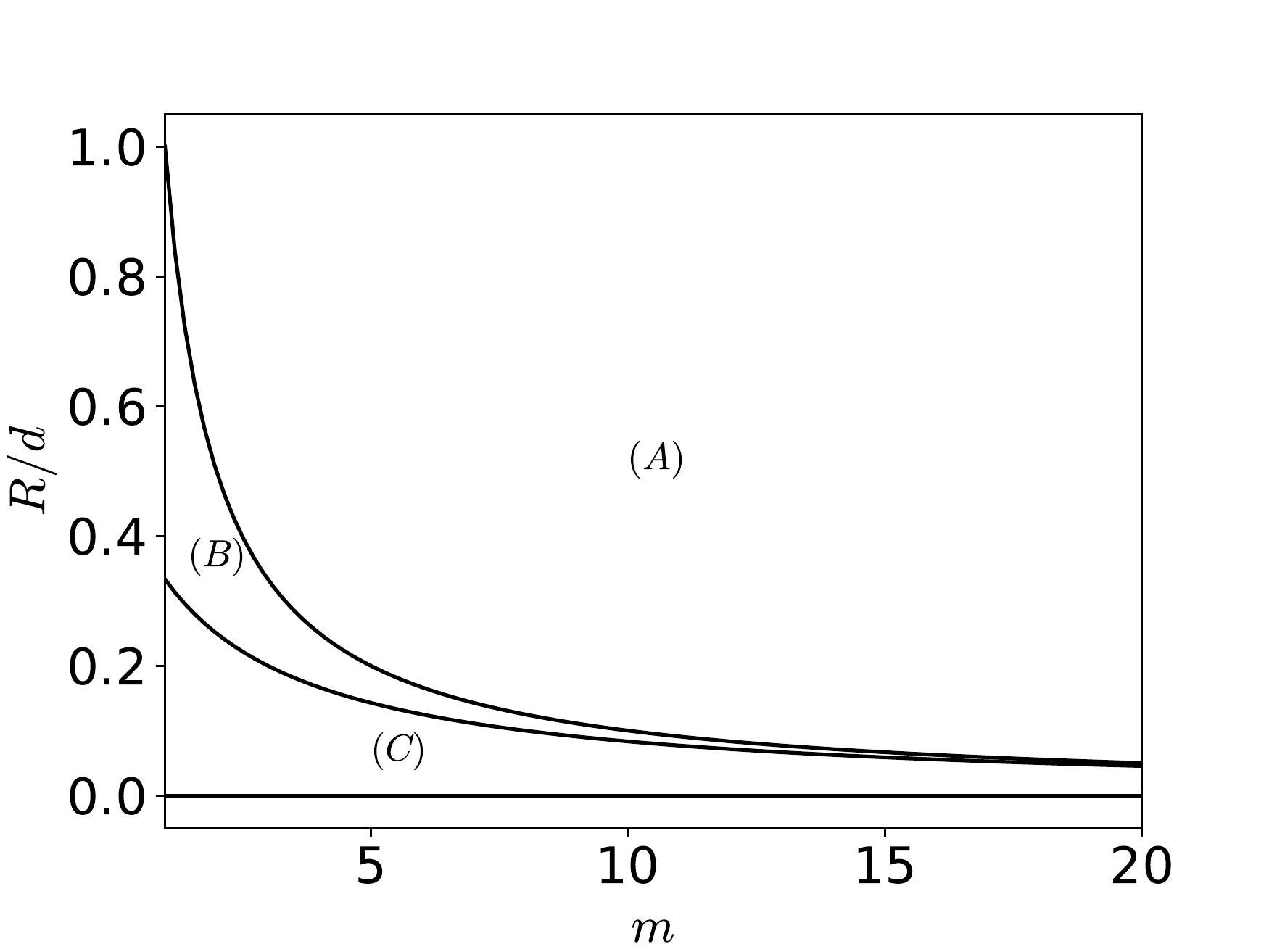}
  \caption{The $m-R/d$ parameter plane where $n$ is fixed to $n = 2$.}
  \label{fig:param_plane}
\end{figure}

Fig.~\ref{fig:param_plane} shows the $m-R/d$ parameter plane with $n = 2$ where the meaning of each region is as follows.
In the region $(A)$, both $x_1 = 1$ and $x_1^{*} < 0$ are asymptotically stable equilibrium points, and the basin of attraction of $x_1 = 1$ is $(0, \ \infty)$, that is, every solution of \eqref{eq:dynamics_x1} starting in $(0, 1]$ converges to $1$.
In the region $(B)$, both $x_1 = 0$ and $x_1 = 1$ are asymptotically stable equilibrium points, and basins of attraction of $x_1 = 0$ and $x_1 = 1$ are $(-\infty, \ x_1^{*})$ and $(x_1^{*}, \ \infty)$, respectively,
that is, every solution of \eqref{eq:dynamics_x1} starting in $[0, \ x_1^{*})$ converges to $0$, and every solution of \eqref{eq:dynamics_x1} starting in $(x_1^{*}, \ 1]$ converges to $1$.
In the region $(C)$, both $x_1 = 0$ and $x_1^{*} > 0$ are asymptotically stable equilibrium points, and the basin of attraction of $x_1 = 0$ is $(-\infty, \ 1)$, that is, every solution of \eqref{eq:dynamics_x1} starting in $[0, \ 1)$ converges to $0$.

\begin{figure}[t]
  \centering
  \includegraphics[clip, width = 7.5cm]{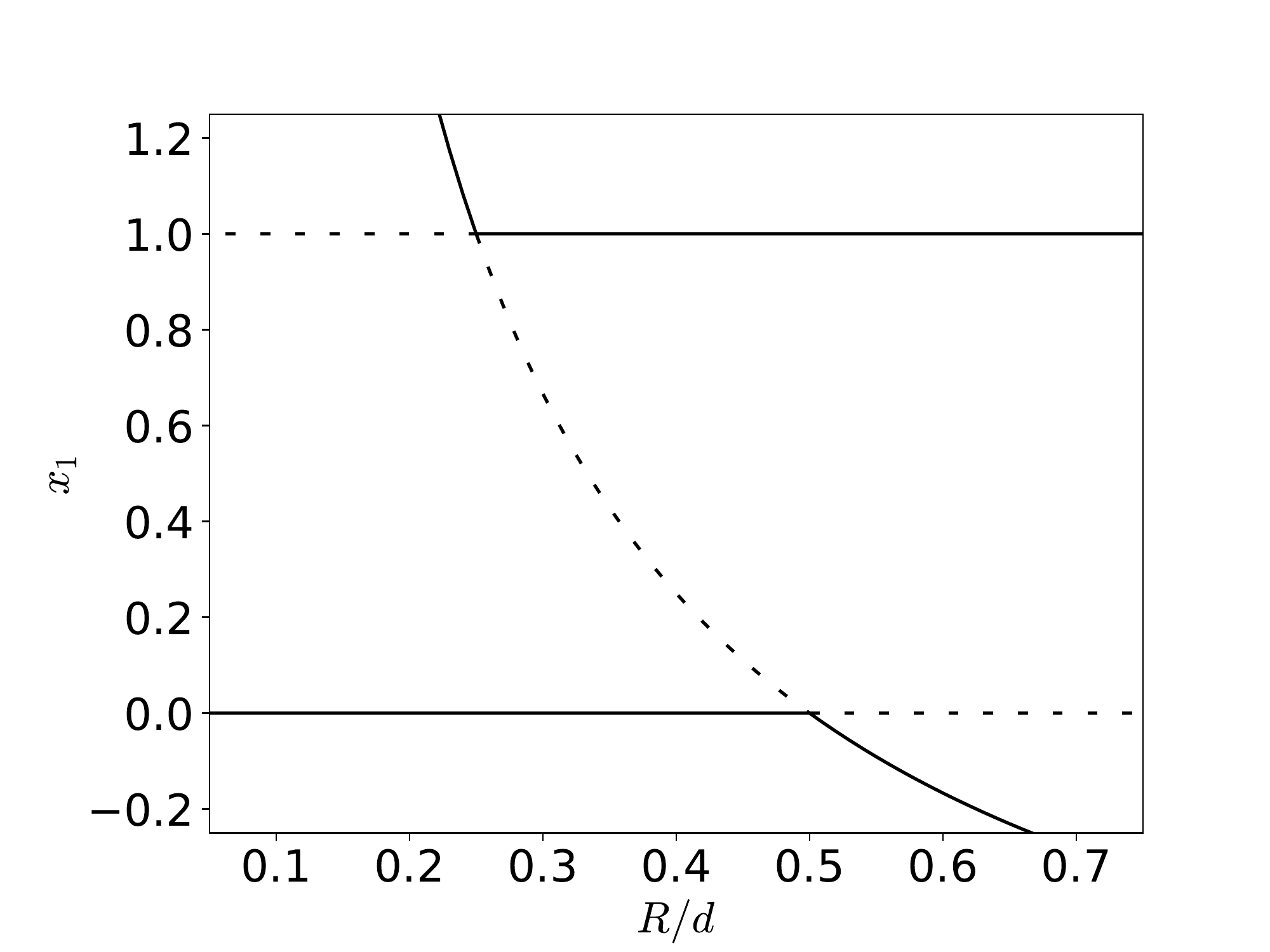}
  \caption{The relation between $R/d$ and the stability of equilibrium points when $m = n = 2$.}
  \label{fig:bif}
\end{figure}

%%%%%%%%%%%%%%%%%%%%%%%%%% modify %%%%%%%%%%%%%%%%%%%%%%%%%%%%%%
Shown in Fig.~\ref{fig:bif} is a bifurcation diagram with respect to the bifurcation parameter $R/d$, where $m = n = 2$.
The solid (\textit{resp.} dashed) line represents an asymptotically stable (\textit{resp.} unstable) equilibrium point.
Two curves of equilibrium points pass through $(x_1, R) = (1, d/(m + n))$ (\textit{resp.} $(x_1, R) = (0, d/m)$), one given by $x_1 = x_1^{*}$, the other by $x_1 = 1$ (\textit{resp.} $x_1 = 0$).
Both curves exist on both sides of $R = d/(m + n)$ (\textit{resp.} $R = d/m$).
The stability along each curve exchanges on passing through $R = d/(m + n)$ (\textit{resp.} $R = d/m$).
Thus, the exchange of stability (known as a \textit{transcritical bifurcation})~\cite{wiggins2003intro} is observed when $(x_1, R) = (0, d/m), (1, d/(m + n))$.
We show in~\ref{sec:Appendix_A} that the vector field \eqref{eq:dynamics_x1} satisfies the condition of the transcritical bifurcation shown in~\cite{wiggins2003intro}.
%%%%%%%%%%%%%%%%%%%%%%%%%%%%%%%%%%%%%%%%%%%%%%%%%%%%%%%%%%%%%%%%%%%

Since the values of $x_i \, (i = 0, 1)$ satisfy $0 \leq x_i \leq 1$, we observe jump phenomena owing to these transcritical bifurcations.
Moreover, $R > d/m$ needs to be satisfied so that all miners in $\mathcal{N}$ participate in the mining.
It is noted that, once the miners participate in the mining, they continue to participate in the mining until the reward $R$ becomes $d/(m + n)$.
Thus, a hysteresis phenomenon of the equilibrium points is observed.

\begin{figure}[t]
  \centering
  \includegraphics[clip, width = 7.5cm]{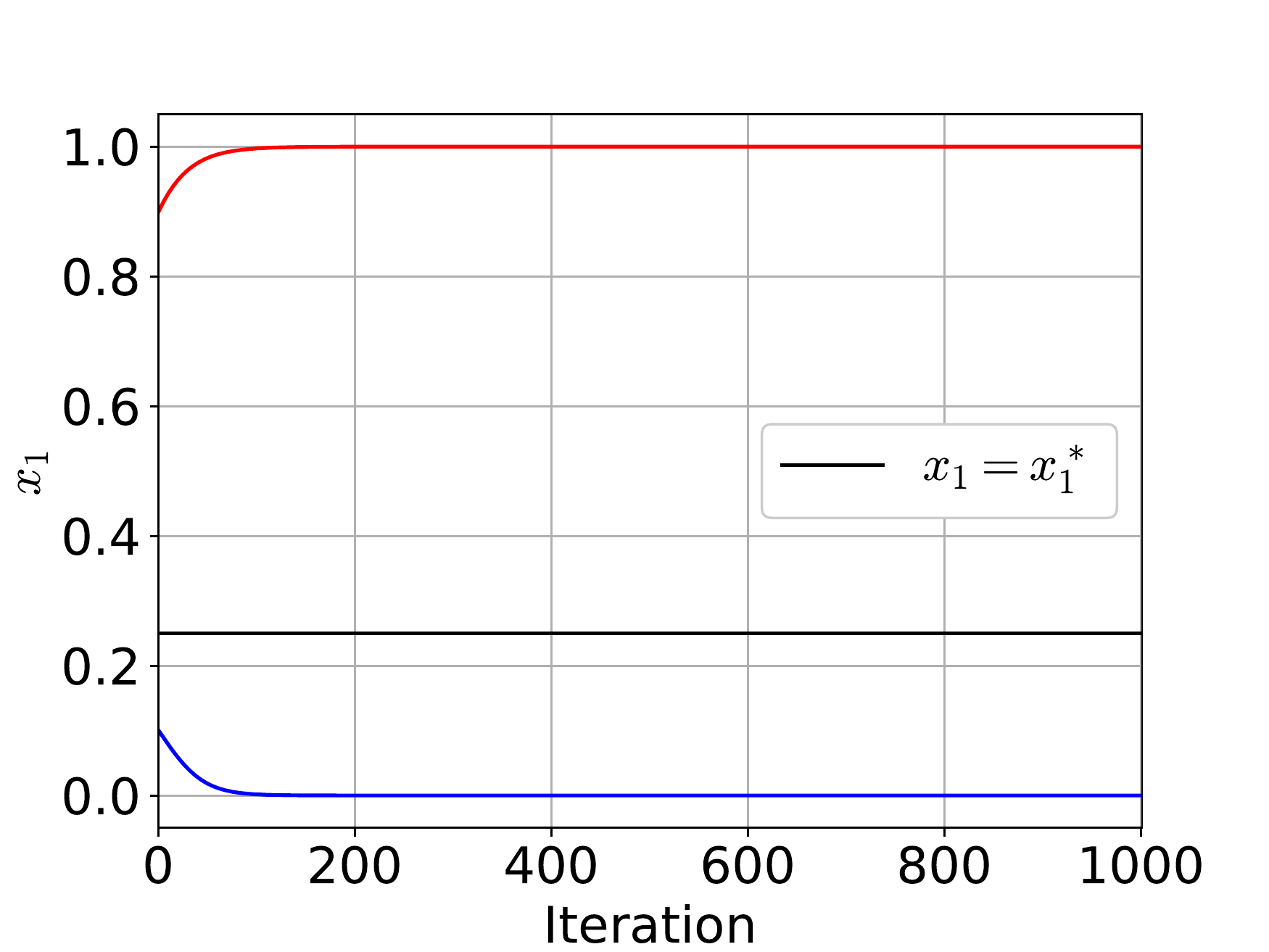}
  \caption{Trajectories of \eqref{eq:dynamics_x1} when $m = n = 2$, $d = 100$, and $R = 40$, from $x_1^{\mathrm{init}} = 0.1$ (blue) and $x_1^{\mathrm{init}} = 0.9$ (red).}
  \label{fig:time_evolution_dynamics}
\end{figure}

Fig.~\ref{fig:time_evolution_dynamics} shows trajectories of \eqref{eq:dynamics_x1} from an initial state $x_1^{\mathrm{init}} = 0.1$ (blue) and $x_1^{\mathrm{init}} = 0.9$ (red).
When $d/(m + n) < R < d/m$, both $x_1 = 0$ and $x_1 = 1$ are asymptotically stable points whose basins of attraction are $(-\infty, \ x_1^{*})$ and $(x_1^{*}, \ \infty)$, respectively.
Thus, the number of miners who participate in the mining converges to $0$ if the initial ratio is less than $x_1^{*}$ because their utility is negative and they prefer non-participation.

\section{Stabilization}\label{sec:stabilize}
The result in Section~\ref{sec:analysis} implies that no miner in $\mathcal{N}$ participates in the mining in the steady state when the mining reward $R^{*}$ satisfies $R^{*} < d/(m + n)$.
When the mining reward $R^{*}$ satisfies $d/(m + n) < R^{*} < d/m$, miners' behaviors depend on the initial state $x_1^{\mathrm{init}}$, i.e., no miner in $\mathcal{N}$ participates in the mining in the steady state when $x_1^{\mathrm{init}} < x_1^{*}$.
We propose a state feedback controller to adjust the reward based on the ratio $x_1$ so that all miners in $\mathcal{N}$ participate in the mining, i.e., let every trajectory of $x_1$ with its initial state in $(0, \ 1]$ converge to $1$.

\subsection{Case where $R^{*} < d/(m + n)$}
First, we show that $x_1 = 1$ cannot be stabilized when $R^{*} < d/(m + n)$.
We consider the following state feedback controller $R_1(x_1)$ that adjusts the reward based on the ratio $x_1$.
\begin{align}
  R = R_1(x_1), \; R_1(1) = R^{*} < \frac{d}{m + n}. \label{eq:controller_1}
\end{align}
The controlled trajectory of $x_1$ by \eqref{eq:controller_1} is described by
\begin{align}
  \dot{x}_1 = \frac{x_1(1 - x_1)}{m + nx_1} \left(R_1(x_1) - \frac{d}{m + n}\right) \eqqcolon \psi_R(x_1). \label{eq:dynamics_control_1}
\end{align}
The derivative of $\psi_R(x_1)$ with respect to $x_1$ is
\begin{align}
  &\frac{\partial \psi_R(x_1)}{\partial x_1} = \left(-\frac{x_1}{m + nx_1} + \frac{1 - x_1}{m + nx_1} - \frac{nx_1(1 - x_1)}{(m + nx_1)^2}\right) \nonumber \\
  &\hspace{15mm}\times \left(R_1(x_1) - \frac{d}{m + nx_1}\right) \nonumber \\
  &\hspace{17mm}+ \frac{x_1(1 - x_1)}{m + nx_1} \left(\frac{\partial R_1(x_1)}{\partial x_1} + \frac{dn}{(m + nx_1)^2}\right). \label{eq:partial_psi}
\end{align}
We obtain
\begin{align}
  \left.\frac{\partial \psi_R(x_1)}{\partial x_1}\right|_{x_1 = 1} = -\frac{1}{m + n} \left(R_1(1) - \frac{d}{m + n}\right) > 0.
\end{align}
Therefore, the unstable equilibrium point $x_1 = 1$ cannot be stabilized even if the feedback controller is used.

\subsection{Case where $d/(m + n) < R^{*} < d/m$}
Next, we show that $x_1 = 1$ can be an asymptotically stable equilibrium point whose basin of attraction is $(0, \ 1]$ with a state feedback controller.
We introduce the following state feedback controller $R_2(x_1)$ to adjust the reward based on the ratio $x_1$,
\begin{align}
  R = R_2(x_1) = R^{*} + \Delta R(x_1), \; \Delta R(1) = 0, \label{eq:controller_2}
\end{align}
and let every trajectory of $x_1$ with its initial state in $(0, \ 1]$ converge to $1$.

\subsubsection{The condition of the feedback gain}
We give $\bar{x}_1$ satisfying $x_1^{*} < \bar{x}_1 \leq 1$ and $\varepsilon > 0$.
For a given reward $R^{*} \in (d/(m + n), d/m)$, let $\Delta R(x_1)$ be
\begin{align}
  \Delta R(x_1) = \begin{cases}
    K(\bar{x}_1 - x_1) & \mbox{if} \; \; x_1 < x_1^{*} + \varepsilon, \\
    0 & \mbox{otherwise},
\end{cases}\label{eq:controller_delta_2}
\end{align}
where $K > 0$ is a feedback gain.
We obtain a condition for the gain $K$ and $\varepsilon$ such that every trajectory of $x_1$ with its initial state in $(0, \ 1]$ converges to $1$ as in Proposition~\ref{prop:controller}.
\begin{prop}\label{prop:controller}
  Assume $d/(m + n) < R^{*} < d/m$.
  Let $\zeta_R(x_1)$ be
  \begin{align}
    \zeta_R(x_1) &\coloneqq -Knx_1^2 + (Kn \bar{x}_1 - Km + R^{*}n)x_1 \nonumber \\
     &\hspace{25mm}+ (R^{*}m + Km\bar{x}_1 - d), \label{eq:zeta_x1}
  \end{align}
  and let $\alpha, \beta \; (\alpha < \beta)$ be real solutions of the quadratic equation $\zeta_R(x_1) = 0$.
  Then, every trajectory of $x_1$ with its initial state in $(0, \ 1]$ converges to $1$ if the gain $K$ and $\varepsilon$ satisfy
  \begin{align}
    &K > \frac{d - R^{*}m}{m \bar{x}_1} \ (> 0), \label{eq:cond_gainK} \\
    &0 < \varepsilon \begin{cases}
     < \beta - x_1^{*} & {\rm if} \; \; \beta < 1, \\
     \leq 1 - x_1^{*} & {\rm if} \; \; \beta \geq 1.
  \end{cases} \label{eq:cond_eps}
  \end{align}
\end{prop}

\begin{proof}
  With the controller~\eqref{eq:controller_2} and \eqref{eq:controller_delta_2}, the dynamics of $x_1$ ($x_1 < x_1^{*} + \varepsilon$) is given by
  \begin{align}
    &\dot{x}_1 = \eta_R(x_1), \label{eq:dynamics_control} \\
    &\eta_R(x_1) \coloneqq \frac{x_1(1 - x_1)}{m + nx_1} \left(R^{*} + K(\bar{x}_1 - x_1) - \frac{d}{m + nx_1}\right). \label{eq:eta}
  \end{align}
  According to \eqref{eq:zeta_x1}, $\eta_R(x_1)$ can be rewritten as
  \begin{align}
    \eta_R(x_1) = \frac{x_1(1 - x_1)\zeta_R(x_1)}{(m + nx_1)^2}. \label{eq:eta_rewrite}
  \end{align}

  First, we prove that the quadratic equation $\zeta_R(x_1) = 0$ has two distinct real solutions under \eqref{eq:cond_gainK}.
  We have
  \begin{align}
    &\zeta_R(x_1^{*}) = K(\bar{x}_1 - x_1^{*})(m + nx_1) > 0, \label{eq:zeta_x1str}
  \end{align}
  which implies with \eqref{eq:cond_gainK} that the quadratic equation $\zeta_R(x_1) = 0$ has two distinct real solutions $\alpha, \beta$ satisfying $\alpha < x_1^{*} < \beta$.

  Next, we prove $\alpha < 0$ under \eqref{eq:cond_gainK}.
  We obtain
  \begin{align}
    \zeta_R(0) &= m\bar{x}_1 K - (d - R^{*}m) \nonumber \\
    &> m\bar{x}_1 \frac{d - R^{*}m}{m\bar{x}_1} - (d - R^{*}m) = 0, \label{eq:zeta_0}
  \end{align}
  from \eqref{eq:zeta_x1} and \eqref{eq:cond_gainK}.
  Since $\zeta_R(x_1)$ is a convex upward quadratic function, the smaller solution $\alpha$ of $\zeta_R(x_1) = 0$ satisfies $\alpha < 0$.

  Finally, we prove that the system controlled by \eqref{eq:controller_2} satisfies $\dot{x}_1 > 0$ for any $x_1 \in (0, \ 1)$ under \eqref{eq:cond_gainK} and \eqref{eq:cond_eps}.
  When $\beta < 1$, $\eta_R(x_1) > 0$ for any $x_1 \in (0, \ \beta)$ from \eqref{eq:eta_rewrite}.
  It is obvious that $\dot{x}_1 > 0$ for any $x_1 \in (0, x_1^{*} + \varepsilon)$ from \eqref{eq:cond_eps}.
  For any $x_1 \in [x_1^{*} + \varepsilon, 1)$, $\dot{x}_1 > 0$ because $K = 0$.
  Thus,
  \begin{align}
    \dot{x}_1 > 0 \; \mbox{for} \; \mbox{any} \;  x_1 \in (0, \; 1). \label{eq:dx1positive}
  \end{align}
  Similary, it is also shown by \eqref{eq:cond_eps} that \eqref{eq:dx1positive} holds for any $\beta \geq 1$.
  Therefore, every trajectory of $x_1$ with its initial state in $(0, \ 1]$ converges to $1$ under \eqref{eq:cond_gainK} and \eqref{eq:cond_eps}.
\end{proof}

It is noted that $\bar{x}_1<\beta$ since $d/(m+n)<R^*$.  So, \eqref{eq:controller_delta_2} is continuous if $\varepsilon=\bar{x}_1-x_{1}^{*}$.

\subsubsection{Performance evaluation}
In this section, we provide the numerical analysis of the controller.
We consider the case where $m = n = 2$, $d = 100$, and $R^{*} = 40$.
Then, we have $x_1^{*} = 0.25$ from \eqref{eq:xstr}.
Let the initial state of $x_1$ be $x_1^{\mathrm{init}} = 0.1$.
We consider the following two cases where $K$ and $\varepsilon$ satisfy \eqref{eq:cond_gainK} and \eqref{eq:cond_eps}.
\begin{description}
  \item[Case 1)] $\bar{x}_1 = 0.26, \; \varepsilon = 0.005, \; K = 56.8125$.
  \item[Case 2)] $\bar{x}_1 = 1, \; \varepsilon = 0.75, \; K = 10.1$.
\end{description}

\begin{figure}[t]
  \begin{minipage}[b]{0.99\linewidth}
    \centering
    \includegraphics[clip, width = 7.5cm]{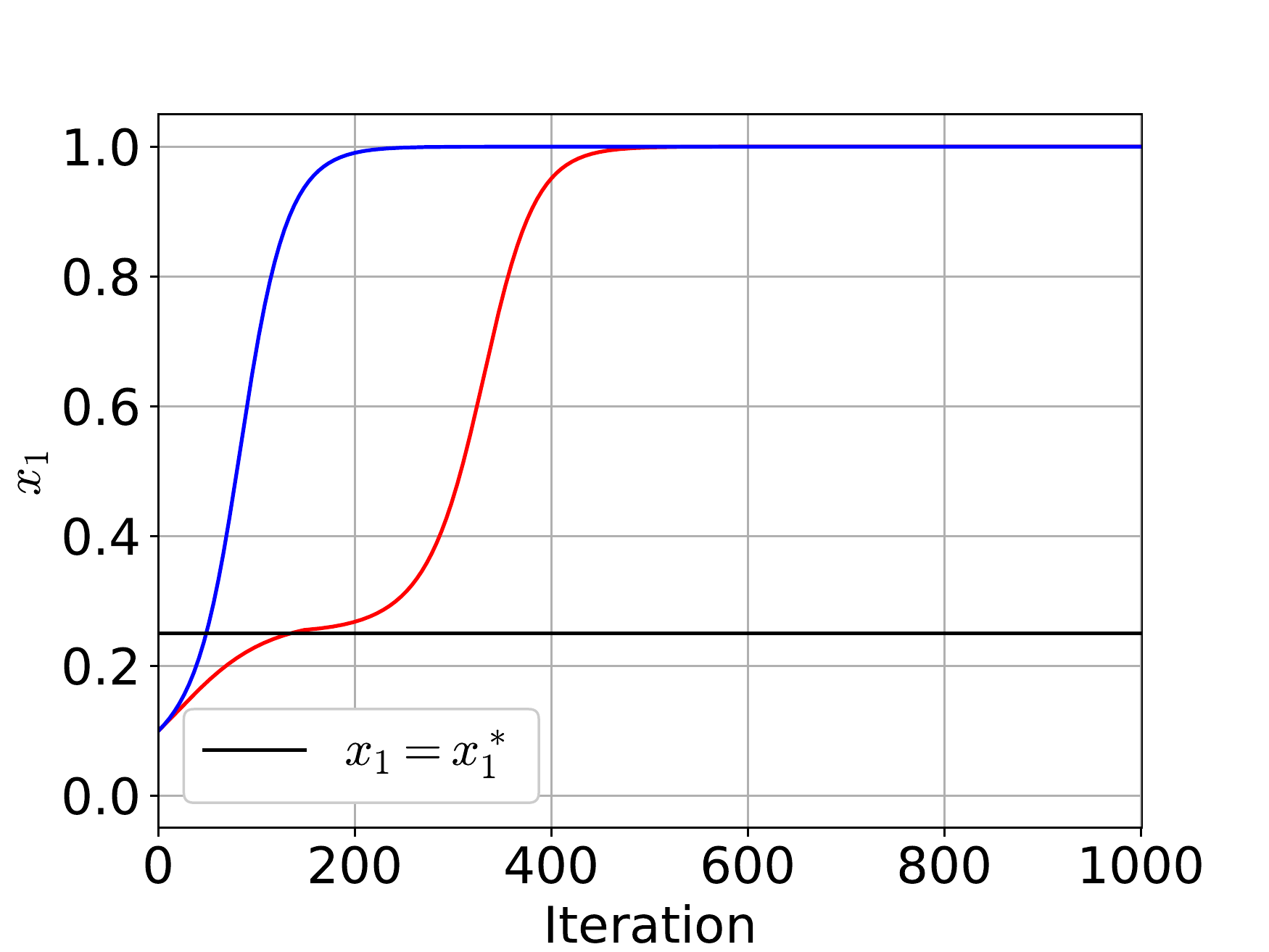}
    \subcaption{}
    \label{fig:trajectory}
  \end{minipage}\\
  \begin{minipage}[b]{0.99\linewidth}
    \centering
    \includegraphics[clip, width = 7.5cm]{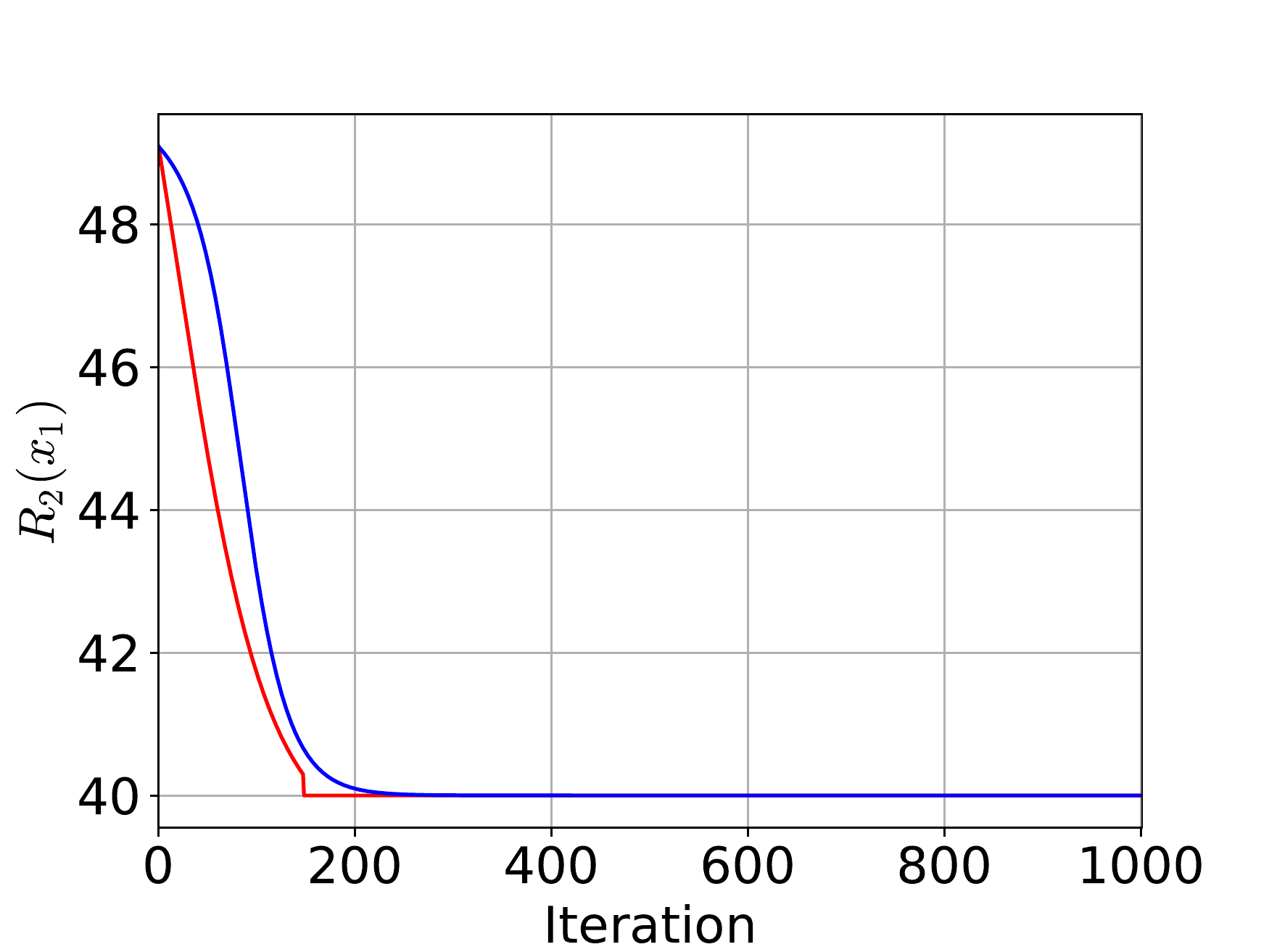}
    \subcaption{}
    \label{fig:with_control_input}
  \end{minipage}
  \caption{Trajectories of (a) the state $x_1$ and (b) the reward $R_2(x_1)$ with  a feedback controller satisfying Proposition~\ref{prop:controller},
  when $m = n = 2$, $d = 100$, $R^{*} = 40$, $x_1^{*} = 0.25$ from $x_1^{\mathrm{init}} = 0.1$, where $\bar{x}_1 = 0.26, \varepsilon = 0.005, K = 56.8125$ (red) and $\bar{x}_1 = 1, \varepsilon = 0.75, K = 10.1$ (blue).}
  \label{fig:analyze_control}
\end{figure}

Fig.~\ref{fig:analyze_control} shows trajectories of the state and the reward.
The red and blue lines represent the trajectories of Cases 1) and 2), respectively.
In Case 1), it takes a longer time than Case 2) for the state $x_1$ to converge to $1$, but the reward $R_2(x_1)$ returns to the original value $R^{*}$ quickly.
Note that $R_2(x_1)$ in Case 1) is not continuous because we switch the input $\Delta R(x_1)$ to $0$ when $x_1 = x_1^{*} + \varepsilon$ (see \eqref{eq:controller_delta_2}).
In Case 2), the state $x_1$ converges to $1$ quickly, but it takes longer time than Case 1) for the reward $R_2(x_1)$ to return to its original value $R^{*}$.
Thus, there is a trade-off in the choice of the design parameters $\bar{x}_1$ and $\varepsilon$.

\section{Conclusion}
We proposed a dynamical model of the decision-making of miners in the blockchain.
The proposed model is described by the 1st-order differential equation.
So, it is simple but its theoretical analysis gives an insight into the characteristics of the decision-making.
We analyzed the stability of its equilibrium points.
We showed the occurrence of the transcritical bifurcations and observed a hysteresis phenomenon.
We also proposed a feedback controller and showed that it can stabilize the state where all miners participate in the mining from any non-zero initial participation ratio of the miners.
Our future work is to extend our model to the case where miners' computational performances are different from each other.

\section*{Acknowledgements}
This research was supported by JST ERATO JPMJER1603.

\appendix
%%%%%%%%%%% modify %%%%%%%%%%%%%%%%%%%%%
\section{Transcritical bifurcation}\label{sec:Appendix_A}
We consider the following system.
\begin{align}
  \dot{x} = f(x, \mu), \; \; x \in \mathbb{R}, \; \; \mu \in \mathbb{R}. \label{eq:A_dynamics}
\end{align}
We assume that
\begin{align}
  f(x, \mu) = x F(x, \mu),  \label{eq:A_dynamics_F}
\end{align}
where $F: \mathbb{R} \times \mathbb{R} \to \mathbb{R}$ satisfies the following condition.
\begin{align}
  F(x, \mu) \coloneqq \begin{cases}
    \frac{f(x, \mu)}{x} & x \neq 0, \\
    \frac{\partial f(0, \mu)}{\partial x} & x = 0.
\end{cases} \label{eq:A_def_F}
\end{align}
Then, it is shown in \cite{wiggins2003intro} that \eqref{eq:A_dynamics} undergoes a transcritical bifurcation at $(x, \mu) = (0, 0)$ if the following three conditions hold.
\begin{description}
  \item[(T1)] $f(0, 0) = 0, \; \frac{\partial f(0, 0)}{\partial x} = 0$,
  \item[(T2)] $\frac{\partial f(0, 0)}{\partial \mu} = 0$,
  \item[(T3)] $\frac{\partial^2 f(0, 0)}{\partial x \partial \mu} \neq 0, \; \frac{\partial^2 f(0, 0)}{\partial x^2} \neq 0$.
\end{description}
Thus we will show that \eqref{eq:dynamics_x1} satisfies the above three conditions at $(x_1, R) = (0, d/m), \ (1, d/(m + n))$.
%We show that \eqref{eq:dynamics_x1} is rewritten in the form of \eqref{eq:A_dynamics_F} and satisfies (T1) -- (T3) when $(x_1, R) = (x_1^{**}, R^{**}) = (0, d/m), (1, d/(m + n))$.

\subsection{Case where $(x_1, R) = (0, d/m)$}\label{appendix:A_1}
First, we consider the following coordination transformation by which $(x_1, R)=(0, d/m)$ is transformed to $(x, \mu)=(0,0)$.
\begin{align}
  \left(
  \begin{array}{c}
    x_1 \\
    R
  \end{array}
  \right) = \left(
  \begin{array}{c}
    x \\
    \mu
  \end{array}
  \right) + \left(
  \begin{array}{c}
    0 \\
    \frac{d}{m}
  \end{array}
  \right). \label{eq:A_transform_1}
\end{align}
Then, we define
\begin{align}
  f(x, \mu) &\coloneqq \varphi_{\mu + \frac{d}{m}}(x) \nonumber \\
  &= \frac{x(1 - x)}{m + nx} \left( \mu + \frac{d}{m} - \frac{d}{m + nx} \right) = xF(x, \mu), \label{eq:A_def_phi_0}
\end{align}
where the function $F$ is defined by
\begin{align}
  F(x, \mu) \coloneqq \frac{1 - x}{m + nx} \left(\mu + \frac{d}{m} - \frac{d}{m + nx}\right). \label{eq:A_def_large_phi}
\end{align}
Then, we have
\begin{align}
  \frac{\partial f(x, \mu)}{\partial x} &= F(x,\mu) + x \frac{\partial F(x, \mu)}{\partial x} \nonumber \\
   &= \left(- \frac{x}{m + nx} + \frac{1 - x}{m + nx} - \frac{nx(1 - x)}{(m + nx)^2} \right) \nonumber \\
  &\hspace{10mm}\times \left(\mu + \frac{d}{m} - \frac{d}{m + nx}\right) + \frac{dnx(1 - x)}{(m + nx)^3}. \label{eq:A_phi_deriv_x}
\end{align}
It is obvious that
\begin{align}
  F(x, \mu) = \frac{f(x, \mu)}{x} \label{eq:A_Phi_0_rewrite_not_zero}
\end{align}
when $x \neq 0$ and
\begin{align}
  F(0, \mu) = \frac{\mu}{m} = \frac{\partial f(0, \mu)}{\partial x} \; \; (\because \eqref{eq:A_phi_deriv_x}) \label{eq:A_Phi_0_rewrite_zero}
\end{align}
when $x = 0$.
Thus, the function $f$ defined by \eqref{eq:A_def_phi_0} satisfies \eqref{eq:A_dynamics_F} and \eqref{eq:A_def_F}.

Next, we show that $f(x, \mu)$ satisfies the conditions (T1) -- (T3).
We obtain
\begin{align}
  &\frac{\partial f(x, \mu)}{\partial \mu} = \frac{x(1 - x)}{m + nx}, \label{eq:A_dphi_R} \\
  &\frac{\partial^2 f(x, \mu)}{\partial x \partial \mu} = -\frac{x}{m + nx} + \frac{1 - x}{m + nx} - \frac{nx(1 - x)}{(m + nx)^2}, \label{eq:A_ddphi_xR} \\
  &\frac{\partial^2 f(x, \mu)}{\partial x^2} = \left(\frac{n^2x(1 - x)}{(m + nx)^2} - \frac{n(1 - x)}{m + nx} + \frac{nx}{m + nx} - 1\right) \nonumber \\
  &\hspace{5mm}\times \frac{2}{m + nx}\left(\mu + \frac{d}{m} - \frac{d}{m + nx}\right) + \frac{2dn}{(m + nx)^3} \nonumber \\
  &\hspace{10mm}\times \left(\frac{-2nx(1 - x)}{m + nx} - x + (1 - x)\right). \label{eq:A_ddphixx}
\end{align}
Thus, $f(x, \mu)$ satisfies the conditions (T1) -- (T3) because
\begin{align}
  f(0, 0) &= 0, \label{eq:A_0_1} \\
  \frac{\partial f(0, 0)}{\partial x} &=  0, \label{eq:A_0_2} \\
  \frac{\partial f(0, 0)}{\partial \mu} &= 0, \label{eq:A_0_3} \\
  \frac{\partial^2 f(0, 0)}{\partial x \partial \mu} &= \frac{1}{m} \neq 0, \label{eq:A_0_4} \\
  \frac{\partial^2 f(0, 0)}{\partial x^2} &= \frac{2dn}{m^3} \neq 0. \label{eq:A_0_5}
\end{align}

\subsection{Case where $(x_1, R) = (1, d/(m + n))$}
It is noted that the dynamics of $x_0$ is written as follows.
\begin{align}
  \dot{x}_0 &= - \frac{x_0(1 - x_0)}{m + n(1 - x_0)} \left(R - \frac{d}{m + n(1 - x_0)}\right) \nonumber \\
  &= -\varphi_R(1 - x_0) \label{eq:A_x_0_dyn}
\end{align}
because $x_0$ and $x_1$ satisfies \eqref{eq:relation_x0x1}.
Thus, it is sufficient to show that \eqref{eq:A_x_0_dyn} undergoes a transcirtical bifurcation at $(x_0, R) = (0, d/(m + n))$.
We consider the following coordination transformation by which $(x_0, R) = (0, d/(m + n))$ is transformed to $(x, \mu) = (0, 0)$.
\begin{align}
  \left(
  \begin{array}{c}
    x_0 \\
    R
  \end{array}
  \right) =
   \left(
  \begin{array}{c}
    x \\
    \mu
  \end{array}
  \right) + \left(
  \begin{array}{c}
    0 \\
    \frac{d}{m + n}
  \end{array}
  \right). \label{eq:A_transform_2}
\end{align}
Then, we define
\begin{align}
  f(x, \mu) &\coloneqq - \varphi_{\mu + \frac{d}{m + n}}(1 - x) \nonumber \\
  &= - \frac{x(1 - x)}{m + n(1 - x)} \nonumber \\
   &\hspace{10mm} \times \left( \mu + \frac{d}{m + n} - \frac{d}{m + n(1 - x)} \right) \nonumber \\
   &= xF(x, \mu), \label{eq:A_phi_1_x_0}
\end{align}
where the function $F$ is defined by
\begin{align}
  F(x, \mu) &\coloneqq - \frac{1 - x}{m + n(1 - x)} \nonumber \\
   &\hspace{10mm} \times \left( \mu + \frac{d}{m + n} - \frac{d}{m + n(1 - x)} \right). \label{eq:A_def_large_phi_1}
\end{align}
Then, we have
\begin{align}
  \frac{\partial f(x, \mu)}{\partial x} &= -\left( \frac{-x + (1 - x)}{m + n(1 -  x)} + \frac{nx(1 - x)}{(m + n(1 - x))^2} \right) \nonumber \\
  &\hspace{5mm}\times \left(\mu + \frac{d}{m + n} - \frac{d}{m + n(1 - x)}\right) \nonumber \\
  &\hspace{10mm}+ \frac{dnx(1 - x)}{(m + n(1 - x))^3}.    \label{eq:A_phi_1_deriv_x}
\end{align}
It is obvious that
\begin{align}
  F(x, \mu) = \frac{f(x, \mu)}{x} \label{eq:A_Phi_1_rewrite_not_zero}
\end{align}
when $x \neq 0$ and
\begin{align}
  F(0, \mu) = -\frac{\mu}{m + n} = \frac{\partial f(0, \mu)}{\partial x} \; \; (\because \eqref{eq:A_phi_1_deriv_x}) \label{eq:A_Phi_1_rewrite_zero}
\end{align}
when $x = 0$.
Thus, the function $f$ defined by \eqref{eq:A_phi_1_x_0} satisfies \eqref{eq:A_dynamics_F} and \eqref{eq:A_def_F}.

\begin{comment}
Next, we show that $f(x, \mu)$ satisfies the condition (T1) -- (T3).
We obtain
\begin{align}
  &\frac{\partial f(x, \mu)}{\partial \mu} = -\frac{x(1 - x)}{m + n(1 - x)}, \label{eq:A_dphi_R_2} \\
  &\frac{\partial^2 f(x, \mu)}{\partial x \partial \mu} = \frac{x - (1 - x)}{m + n(1 -  x)} - \frac{nx(1 - x)}{(m + n(1 - x))^2}, \label{eq:A_ddphi_xR_2} \\
  &\frac{\partial^2 f(x, \mu)}{\partial x^2} = \left(\frac{-n^2x(1 - x)}{(m + n(1 - x))^2} - \frac{n((1 - x) - x)}{m + n(1 - x)} + 1\right) \nonumber \\
  &\hspace{3mm}\times \frac{2}{m + n(1 - x)}\left(\mu + \frac{d}{m + n} - \frac{d}{m + n(1 - x)}\right) \nonumber \\
  &\hspace{6mm}+ \frac{2dn}{(m + n(1 - x))^3} \left(\frac{2nx(1 - x)}{m + n(1 - x)} - x + (1 - x)\right). \label{eq:A_ddphixx_2}
\end{align}
\end{comment}
In the same way as Appendix~\ref{appendix:A_1}, we obtain the partial derivatives of $f(x, \mu)$ and show that $f(x, \mu)$ defined by \eqref{eq:A_x_0_dyn} satisfies the conditions (T1) -- (T3) because
\begin{align}
  f(0, 0) &= 0, \label{eq:A_A_1} \\
  \frac{\partial f(0, 0)}{\partial x} &= 0, \label{eq:A_A_2} \\
  \frac{\partial f(0, 0)}{\partial \mu} &= 0, \label{eq:A_A_3} \\
  \frac{\partial^2 f(0, 0)}{\partial x \partial \mu} &= -\frac{1}{m + n} \neq 0, \label{eq:A_A_4} \\
  \frac{\partial^2 f(0, 0)}{\partial x^2} &= \frac{2dn}{(m + n)^3} \neq 0. \label{eq:A_A_5}
\end{align}

Therefore, \eqref{eq:dynamics_x1} undergoes transcritical bifurcations at $(x_1, R) = (0, d/m), \ (1, d/(m + n))$.

\bibliographystyle{IEEEtran}
\bibliography{ref}

% Generated by IEEEtran.bst, version: 1.14 (2015/08/26)
\begin{thebibliography}{10}
\providecommand{\url}[1]{#1}
\csname url@samestyle\endcsname
\providecommand{\newblock}{\relax}
\providecommand{\bibinfo}[2]{#2}
\providecommand{\BIBentrySTDinterwordspacing}{\spaceskip=0pt\relax}
\providecommand{\BIBentryALTinterwordstretchfactor}{4}
\providecommand{\BIBentryALTinterwordspacing}{\spaceskip=\fontdimen2\font plus
\BIBentryALTinterwordstretchfactor\fontdimen3\font minus
  \fontdimen4\font\relax}
\providecommand{\BIBforeignlanguage}[2]{{%
\expandafter\ifx\csname l@#1\endcsname\relax
\typeout{** WARNING: IEEEtran.bst: No hyphenation pattern has been}%
\typeout{** loaded for the language `#1'. Using the pattern for}%
\typeout{** the default language instead.}%
\else
\language=\csname l@#1\endcsname
\fi
#2}}
\providecommand{\BIBdecl}{\relax}
\BIBdecl

\bibitem{nakamoto2008bitcoin}
\BIBentryALTinterwordspacing
S.~Nakamoto, ``Bitcoin: A peer-to-peer electronic cash system,'' 2008.
  [Online]. Available: \url{http://bitcoin.org/bitcoin.pdf}
\BIBentrySTDinterwordspacing

\bibitem{Xia2017}
Q.~{Xia}, E.~B. {Sifah}, K.~O. {Asamoah}, J.~{Gao}, X.~{Du}, and M.~{Guizani},
  ``{MeDShare}: Trust-less medical data sharing among cloud service providers
  via blockchain,'' \emph{IEEE Access}, vol.~5, pp. 14\,757--14\,767, 2017.

\bibitem{Ouaddah2016}
A.~Ouaddah, A.~A.~Elkalam, and A.~A.~Ouahman, ``Fairaccess: a new
  blockchain-based access control framework for the internet of things,''
  \emph{Security and Communication Networks}, vol.~9, no.~18, pp. 5943--5964,
  2016.

\bibitem{TRUBY2018399}
J.~Truby, ``Decarbonizing bitcoin: Law and policy choices for reducing the
  energy consumption of blockchain technologies and digital currencies,''
  \emph{Energy Research \& Social Science}, vol.~44, pp. 399--410, 2018.

\bibitem{Cambridge}
\BIBentryALTinterwordspacing
``Cambridge bitcoin electricity consumption index,'' (accessed on 1 March
  2021). [Online]. Available: \url{https://www.cbeci.org/}
\BIBentrySTDinterwordspacing

\bibitem{Cai2018}
W.~{Cai}, Z.~{Wang}, J.~B. {Ernst}, Z.~{Hong}, C.~{Feng}, and V.~C.~M. {Leung},
  ``Decentralized applications: The blockchain-empowered software system,''
  \emph{IEEE Access}, vol.~6, pp. 53\,019--53\,033, 2018.

\bibitem{liu2021social}
Y.~Liu, Z.~Fang, M.~H. Cheung, W.~Cai, and J.~Huang, ``A social welfare
  maximization mechanism for blockchain storage,'' \emph{arXiv preprint
  arXiv:2103.05866}, 2021.

\bibitem{Liu2019}
Z.~{Liu}, N.~C. {Luong}, W.~{Wang}, D.~{Niyato}, P.~{Wang}, Y.~{Liang}, and
  D.~I. {Kim}, ``{A} survey on blockchain: A game theoretical perspective,''
  \emph{IEEE Access}, vol.~7, pp. 47\,615--47\,643, 2019.

\bibitem{Dimitri2017}
N.~Dimitri, ``Bitcoin mining as a contest,'' \emph{Ledger}, vol.~2, pp. 31--37,
  2017.

\bibitem{Fiat2019}
A.~Fiat, A.~Karlin, E.~Koutsoupias, and C.~Papadimitriou, ``Energy equilibria
  in proof-of-work mining,'' in \emph{Proceedings of the 2019 ACM Conference on
  Economics and Computation}, 2019, pp. 489--502.

\bibitem{weibull1997evolutionary}
J.~W. Weibull, \emph{Evolutionary game theory}.\hskip 1em plus 0.5em minus
  0.4em\relax MIT press, 1997.

\bibitem{Kanazawa2008}
T.~Kanazawa, H.~Goto, and T.~Ushio, ``Replicator dynamics with dynamic payoff
  reallocation based on the government's payoff,'' \emph{IEICE Transactions on
  Fundamentals of Electronics, Communications and Computer Sciences}, vol.
  E91-A, no.~9, pp. 2411--2418, 2008.

\bibitem{Kanazawa2009}
T.~Kanazawa, Y.~Fukumoto, T.~Ushio, and T.~Misaka, ``Replicator dynamics with
  pigovian subsidy and capitation tax,'' \emph{Nonlinear Analysis, Theory,
  Methods and Applications}, vol.~71, no.~12, pp. e818--e826, 2009.

\bibitem{Morimoto2016}
T.~Morimoto, T.~Kanazawa, and T.~Ushio, ``Subsidy-based control of
  heterogeneous multiagent systems modeled by replicator dynamics,'' \emph{IEEE
  Transactions on Automatic Control}, vol.~61, no.~10, pp. 3158--3163, 2016.

\bibitem{Liu2017}
X.~{Liu}, W.~{Wang}, D.~{Niyato}, N.~{Zhao}, and P.~{Wang}, ``Evolutionary game
  for mining pool selection in blockchain networks,'' \emph{IEEE Wireless
  Communications Letters}, vol.~7, no.~5, pp. 760--763, 2018.

\bibitem{Fujita2020}
K.~Fujita, Y.~Zhang, M.~Sasabe, and S.~Kasahara, ``Mining pool selection
  problem in the presence of block withholding attack,'' in \emph{Proceedings
  of 2020 IEEE International Conference on Blockchain}, 2020, pp. 321--326.

\bibitem{Kim2019}
S.~Kim and S.~G. Hahn, ``Mining pool manipulation in blockchain network over
  evolutionary block withholding attack,'' \emph{IEEE Access}, vol.~7, pp.
  144\,230--144\,244, 2019.

\bibitem{Toda2021}
K.~{Toda}, N.~{Kuze}, and T.~{Ushio}, ``Game-theoretic approach to a
  decision-making problem for blockchain mining,'' \emph{IEEE Control Systems
  Letters}, vol.~5, no.~5, pp. 1783--1788, 2021.

\bibitem{Debus2017}
\BIBentryALTinterwordspacing
J.~Debus, ``Consensus methods in blockchain systems,'' \emph{FSBC Working
  Paper}, 2017. [Online]. Available: \url{http://www.fs-blockchain.de/}
\BIBentrySTDinterwordspacing

\bibitem{Houy2016}
N.~Houy, ``The bitcoin mining game,'' \emph{Ledger}, vol.~1, pp. 53--68, 2016.

\bibitem{Kraft2016}
D.~Kraft, ``{Difficulty control for blockchain-based consensus systems},''
  \emph{Peer-to-Peer Networking and Applications}, vol.~9, no.~2, pp. 397--413,
  2016.

\bibitem{1547-5816_2019_1_365}
S.~Kasahara and J.~Kawahara, ``Effect of bitcoin fee on
  transaction-confirmation process,'' \emph{Journal of Industrial \& Management
  Optimization}, vol.~15, no.~1, pp. 365--386, 2019.

\bibitem{khalil2002nonlinear}
H.~K. Khalil, \emph{Nonlinear systems}, 3rd~ed.\hskip 1em plus 0.5em minus
  0.4em\relax Prentice Hall, 2002.

\bibitem{wiggins2003intro}
S.~Wiggins, \emph{Introduction to applied nonlinear dynamical systems and
  chaos}, 2nd~ed.\hskip 1em plus 0.5em minus 0.4em\relax Springer, 2003.

\end{thebibliography}

\end{document}